\documentclass[letterpaper, 10 pt, conference]{ieeeconf}  

\IEEEoverridecommandlockouts                              

\usepackage{graphics} 
\usepackage{epsfig} 
\usepackage{amsmath}
\usepackage{amssymb}
\usepackage{subcaption}
\usepackage{amsfonts}
\usepackage{savesym}
\savesymbol{AND}
 \usepackage{algorithm,algorithmic}
\usepackage{breqn}
\usepackage{tabularx}
\usepackage{xcolor}
\usepackage{multirow}

\newtheorem{nntheorem}{\bf Theorem}[section]
\newtheorem{nnlemma}[nntheorem]{\bf Lemma}
\newtheorem{nndefinition}[nntheorem]{\bf Definition}
\newtheorem{nnproposition}[nntheorem]{\bf Proposition}
\newtheorem{nnassumption}[nntheorem]{\bf Assumption}

\newenvironment{theorem}
{\begin{nntheorem}\it}
{\end{nntheorem}}

\newenvironment{proposition}
{\begin{nnproposition}\it}
{\end{nnproposition}}

\newenvironment{lemma}
{\begin{nnlemma}\it}
{\end{nnlemma}}

\newenvironment{definition}
{\begin{nndefinition}\it}
{\end{nndefinition}}

\newcommand{\V}[1]{\mathbf{#1}}

\newcommand{\vZero}{\V{0}}

\newcommand{\vB}{\V{B}}

\newcommand{\vI}{\V{I}}

\newcommand{\vP}{\V{P}}
\newcommand{\vx}{\V{x}}
\newcommand{\vA}{\V{A}}
\newcommand{\vE}{\V{E}}

\newcommand{\vC}{\V{C}}

\newcommand{\vy}{\V{y}}

\newcommand{\vD}{\V{D}}

\title{\LARGE \bf
Cyber-Physical Security of Vehicles:\\ Zero Dynamics Attacks Against Vehicle's Lateral Dynamics}

\author{Ghadeer Shaaban, Hassen Fourati, Alain Kibangou, Christophe Prieur, and Mohammad Pirani
\thanks{Ghadeer Shaaban, Hassen Fourati, Alain Kibangou, and Christophe Prieur are with Univ. Grenoble Alpes, CNRS, Inria, Grenoble-INP, GIPSA-lab, Grenoble, France. Alain Kibangou is also with
the Faculty of Science, University of Johannesburg (Auckland Park
Campus), Johannesburg 2006, South Africa. Mohammad Pirani is with the Department of Mechanical Engineering,
University of Ottawa, Ottawa, ON K1N 6N5, Canada. (e-mails: ghadeer.shaaban@gipsa-lab.fr; hassen.fourati@gipsa-lab.fr; alain.kibangou@gipsa-lab.fr; christophe.prieur@gipsa-lab.fr; mpirani@uottawa.ca)}%
}

\begin{document}

\maketitle
\thispagestyle{empty}
\pagestyle{empty}

\begin{abstract}
Modern vehicles have evolved from mechanical systems to complex and connected ones controlled by numerous digital computers interconnected through internal networks. While this development has improved their efficiency and safety, it also brings new potential risks, particularly cyber-attacks. Several studies have explored the security of vehicle dynamics against such threats. Among these dynamics, the vehicle's lateral dynamics are crucial for maintaining stability and control during turns and maneuvers, making them a key focus of research. However, only a few recent studies have specifically investigated the security of lateral dynamics.
This paper explores the potential for zero dynamics attacks on the vehicle's lateral dynamics, where the attacker can remain undetected by leaving no trace on the system's outputs. Three scenarios are studied: when the output includes yaw rate, lateral acceleration, and their combination. These two critical measurements of a vehicle's lateral motion are accessible through the inertial measurement units (IMU) in every vehicle. For each scenario, the impact of zero dynamics attacks on system performance is analyzed and illustrated through simulations. Finally, the paper provides recommendations for securing vehicles' lateral dynamics against such attacks.
\end{abstract}

\section{Introduction}
\subsection{Motivation}
In recent years, modern vehicles become complex systems that contain hundreds of electronic control units, actuators and sensors communicating with each other through internal networks.
While this advancement enables significant functionalities and efficiencies, it also makes the vehicle vulnerable to security weaknesses and opens the door for cyber-attacks. The attacker can gain access to the vehicle's internal network, eavesdrop on the messages, and compromise sensor data and control input signals transmitted within it \cite{10.5555/2028067.2028073,kang_attack_2020}. Several cyber-physical attacks against vehicles have occurred in real-world scenarios. For example, an attacker remotely crashed a Jeep from 10 miles away, and another attacker took remote control of a Tesla from 12 miles away \cite{liuSecurePoseEstimation2019a}. As vulnerabilities to cyber-attacks against vehicles threaten human safety, addressing the vehicle’s security becomes a critical concern and fundamental problem that requires dedicated research efforts from both academic and industrial domains. 
\subsection{Literature review}
The research on Cyber-Physical Systems (CPS) security focuses on designing and 
defending systems against one of the three main types of cyber attacks: denial of service attacks (DoS), replay attacks, and
false data injection (FDI) attacks \cite{chongTutorialIntroductionSecurity2019}. In DoS attacks, the attacker blocks the transmission
of sensor measurements and control input signals. In replay attacks, the 
attacker records sensor measurements and replays them later instead of the actual 
measurements. In FDI attacks, the attacker injects
false data into the true sensor measurements and control inputs. 
Some researchs have already studied securing vehicles against
DoS attacks~\cite{abdollahi_biron_real-time_2018}, and reply attacks~\cite{al-shareeda_review_2020}. In~\cite{biroon_false_2022,ju_survey_2022,sun_survey_2022}, the FDI attacks against connected vehicles are studied, where the attacks occur on the transmitted signals between vehicles, rather than within the individual vehicle. In  \cite{hong_esp_2022,peng_security_2024,bertoni_non-invasive_2013}, FDI and spoofing attacks against vehicle sensors attacks are studied. 
\par
Fundamental vehicle dynamics are the lateral dynamics, which describe the vehicle's lateral movement. These dynamics have been widely studied in academic and industrial domains~\cite{gillespie_funamentals_2021,rajamani_vehicle_2012}. Lateral dynamics involve the lateral velocity and yaw rate dynamics, with two inputs: the steering angle, which is an input by the driver, and the yaw moment, which can be generated by independent in-wheel motors. The yaw moment plays a vital role in enhancing vehicle stability and controllability. The yaw rate can be directly measured by IMU sensors, but lateral velocity lacks direct measurement. Instead, it is estimated based on the dynamic model, inputs and other sensors' measurements, such as yaw rate measurements~\cite{kiencke_observation_1997,nam_estimation_2013}, lateral acceleration measurements~\cite{dixon_extended_2000}, or a combination of both~\cite{cheli_methodology_2007,chen_sideslip_2008}. Due to the importance of lateral dynamics, extensive research has focused on control methods \cite{changResilientControlDesign2019,haipingduStabilizingVehicleLateral2010,zhangVehicleLateralDynamics2016,zhangRobustGainschedulingEnergytopeak2014}, as well as on estimation and observation techniques for the lateral model \cite{dixon_extended_2000,doumiatiOnboardRealTimeEstimation2011,kiencke_observation_1997,nam_estimation_2013}.  On the other hand, only few and recent works have focused on the security of lateral model. 
Attacks that modify sensor signals
to cause damage in the lateral control have been proposed in~\cite{farivarCovertAttacksAdversarial2021}. 
In \cite{nekouei_randomized_2022}, security measures are proposed to protect against attackers who aim to infer the values of lateral controller gains.
In \cite{mohammadi_vehicle_2022,mohammadi_vehicle_2023}, the lateral model is compromised by attacking the braking system and continuously varying the longitudinal slip of the wheels.
\par
One class of FDI attacks that target system inputs is the zero dynamics attacks, where the attacker exploits the invariant zeros of the system to perform attacks leaving no trace on the system's outputs, making these attacks undetectable~\cite{pasqualettiAttackDetectionIdentification2013,pasqualettiControlTheoreticMethodsCyberphysical2015}. It is proven in~\cite{Shim2022} that the zero dynamics attacks are disruptive, i.e. puts the system on high risk, if the attacks excite unstable invariant zeros of the system.
\subsection{Contributions}
To the best of the authors' knowledge, zero dynamics attacks have not been studied for vehicle lateral dynamics. Moreover, studies on zero dynamics attacks offer theoretical insights but rarely provide comprehensive, real-world examples. In the current work, we study and analyze the invariant zeros of the vehicle's lateral model, and show how the attacker can exploit these zero dynamics to perform undetectable attacks, leaving no trace to the system's outputs, namely lateral acceleration and yaw rate. We study three cases of output, when the output consists only of yaw rate measurements, when it consists only of lateral acceleration measurements, and when it consists of a combination of both. Additionally, we exploit the relationship between the system's invariant zeros and its strong observability and detectability properties to analyze these characteristics in the lateral dynamics model. 
Our motivation for this work is not to create zero dynamics attacks but to evaluate vehicle security against them and improve protection measures.
 The main contributions of this work are:
\begin{enumerate}
\item \textbf{Attacks Design:} Study the existence of invariant zeros of the vehicle's lateral dynamics, design zero dynamics attacks and explore their potential to be disruptive.
    \item \textbf{Attacks Prevention:} Suggest measures to protect the vehicle's lateral dynamics against zero dynamics attacks.
    \item \textbf{Observability Under Attacks:} Investigate the strong observability and detectability properties of the lateral dynamics model by examining its invariant zeros.
\end{enumerate}
\par 
 The remainder of this paper is as follows: Section~\ref{sec:Preliminaries and Problem formulation} provides preliminaries on the vehicle lateral model, invariant zeros, zero dynamics attacks, and the problem statement. Section~\ref{sec:Zero Dynamic attack against vehicle lateral dynamic} studies the existence of invariant zeros, designs and analyzes the zero dynamics attacks for the three cases of output. Section~\ref{sec:simulation} provides some simulations to illustrate the findings. Finally, Section~\ref{sec:conclusion} concludes the paper.
\section{Preliminaries and Problem statement}
This section covers the fundamentals of vehicle's lateral dynamics, as well as the concepts of zero dynamics and zero dynamics attacks.
\label{sec:Preliminaries and Problem formulation}
\subsection{Vehicle's linear lateral model}
The two-degrees-of-freedom \textit{bicycle model} is a widely used approach for analyzing vehicle lateral dynamics, describing the dynamics of lateral velocity, $v_y$, and yaw rate, $r=\dot{\psi}$ \cite{gillespie_funamentals_2021,rajamani_vehicle_2012}, where $\psi$ is the vehicle's yaw (heading) angle. The vehicle's bicycle model with front-wheel steering is shown in Fig.~\ref{fig: Plan-view-vehicle}. The body frame which is attached to the vehicle, has its origin in the vehicle's center of gravity (CG), its $x$-axis and $y$-axis are aligned with the longitudinal and lateral axes, respectively. The distances from CG to the front and rear axles are denoted by $a$ and $b$, respectively. The steering angle, which is a command by the driver is denoted by $\delta$. In addition to the steering angle $\delta$, a control input $M_z$, representing the yaw moment, is designed to stabilize the lateral motion. 
\begin{figure}[hbt]
    \centering
    \includegraphics[width=0.99\linewidth]{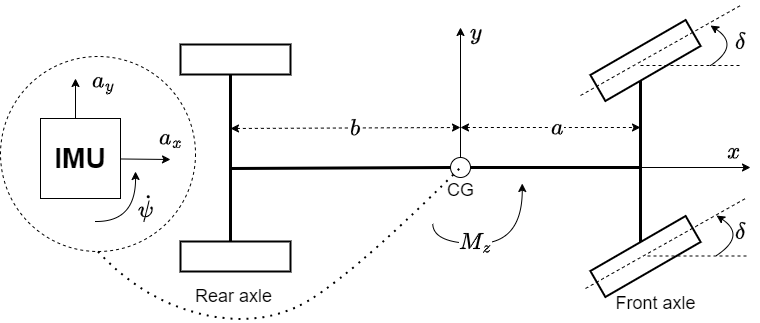}
    \caption{Plan view of vehicle dynamics model.}
    \label{fig: Plan-view-vehicle}
\end{figure}
The following linear state space model describes the vehicle lateral motion dynamics~\cite{gillespie_funamentals_2021,rajamani_vehicle_2012}:
\begin{equation}
    \dot{\vx}=\vA \vx + \vB M_z + \vE \delta,
    \label{eq: lateral dynamic model}
\end{equation}
where $\vx=\left[v_y\ r \right]^T$ is the state vector, and the matrices $\vA$, $\vB$ and $\vE$ are the state matrix, the input matrix for the yaw moment, and the input matrix for the steering angle, respectively, and have the following forms:
\begin{equation}
    \vA=\left[\begin{array}{cc}
        a_{11} & a_{12} \\
        a_{21} &  a_{22}
    \end{array} \right], \
    \vB=\left[\begin{array}{c}
         0 \\
         b_2 
    \end{array} \right], \
    \vE=\left[\begin{array}{c}
         e_1  \\
         e_2
    \end{array} \right],
\end{equation}
and ${a_{11}=-2 \frac{C_f+C_r}{v_x m}}$, ${a_{12}= 2 \frac{b C_r-a C_f}{v_x m}-v_x}$, ${a_{21}=2 \frac{b C_r-a C_f}{v_x I_z}}$, ${a_{22}=-2 \frac{a^2 C_f+b^2 C_r}{v_x I_z}}$, ${b_2=\frac{1}{I_z}}$, ${e_1=\frac{2 C_f}{m}}$ and ${e_2=\frac{2 a C_f}{I_z}}$. The parameters $C_f$, $C_r$, $v_x$, $m$, $I_z$ are the front and rear wheels' cornering stiffness, the vehicle's longitudinal velocity, the vehicle's total mass, and the vehicle's moment of inertia around $z$-axis, respectively.

\vspace*{5pt}

\textbf{Remark:} Although the bicycle model is relatively simple and relies on certain assumptions to yield a linear representation of the system, it effectively captures key lateral vehicle dynamics. Its effectiveness in practical applications validates its use in control design and analysis~\cite{gillespie_funamentals_2021,rajamani_vehicle_2012}.

\vspace*{5pt}

There is no direct measurement of the vehicle's lateral velocity, and it is
estimated using the dynamic model~\eqref{eq: lateral dynamic model} and knowledge of the system inputs $M_z$ and $\delta$, and measurement of yaw rate~\cite{kiencke_observation_1997,nam_estimation_2013}, lateral acceleration~\cite{dixon_extended_2000}, or both~\cite{cheli_methodology_2007,chen_sideslip_2008}. These measurements are feasible by using an IMU. We describe hereafter the output model in each case of measurements.
The first case concerns yaw rate as output, thus $y_1=r$, the output matrix is:
\begin{equation}
  \vC_1=\left[0 \ 1\right],
\label{eq:C1}
\end{equation}
The second case concerns lateral acceleration as output, thus $y_2=a_y=\dot{v}_y+v_x r$. From the dynamic model~\eqref{eq: lateral dynamic model} we have $y_2=a_{11}v_y+a_{12}r+e_1\delta+v_xr=\left[\begin{array}{cc}
      a_{11}   & a_{12}+v_x
    \end{array}\right]\vx + e_1 \delta$, the output matrix is: 
\begin{equation}
    \vC_2=\left[\begin{array}{cc}
      a_{11}   & a_{12}+v_x
    \end{array}  \right].
    \label{eq:C2}
\end{equation}
In the third case, the output consists of the two measurements, and 
the output model is given by $\vy_3=\left[r \ \ a_y \right]^T=\vC_3 \vx + \left[0 \ \ e_1\right]^T\delta$, where the output matrix $\vC_3$ is:
\begin{equation}
    \vC_3=\left[ 
    \begin{array}{cc}
        0 & 1 \\
        a_{11}  & a_{12}+v_x
    \end{array}
    \right].
    \label{eq:C3}
\end{equation}
The following equation summarizes the output model for the three cases of measurements:
\begin{equation}
  \vy_i=\vC_i\vx+\vD_i^\delta \delta,  
  \label{eq:output-model}
\end{equation}
where $i\in \{1,2,3\}$, $\vD_1^\delta=0$,  $\vD_2^\delta=e_1$, and $\vD_3^\delta=[0\ \  e_1]^T$.

\subsection{Zero dynamics attacks}
This subsection defines the undetectable attack~\cite{pasqualettiAttackDetectionIdentification2013,pasqualettiControlTheoreticMethodsCyberphysical2015}, and its relation with the existence of invariant zeros. It further presents the connection between zero dynamics attacks and strong observability and detectability properties. 
\subsubsection{Undetectable zero dynamics attacks}
We adapt the definition of undetectable attacks which is defined in~\cite{pasqualettiAttackDetectionIdentification2013,pasqualettiControlTheoreticMethodsCyberphysical2015} for the vehicle lateral model.
Recalling the linear vehicle lateral dynamics~\eqref{eq: lateral dynamic model} and the output model~\eqref{eq:output-model}:
\begin{equation}
    \begin{aligned}
    \dot{\vx}&=\vA \vx + \vB M_z + \vE \delta,\\
    \vy_i&=\vC_i\vx+\vD_i^\delta \delta.
    \end{aligned}
    \label{eq:lateral-dynamic-output-model}
\end{equation}

 We assume that the attacker injects signal $M_z^a$ to the input $M_z$, without altering the input $\delta$ as it is a mechanical command by the driver. 
\begin{definition}[Undetectable attacks~\cite{pasqualettiAttackDetectionIdentification2013,pasqualettiControlTheoreticMethodsCyberphysical2015}] Consider the linear system~\eqref{eq:lateral-dynamic-output-model}, and let $\vy_i\left(\vx^{\alpha}(t_0),M_z^a,t\right)$ be the system's output at time $t\geq t_0$, given an initial state $\vx^{\alpha}(t_0)$ and the presence of an injection of attack signal $M_z^a$, the attacks are considered undetectable if there exists an initial state $\vx^{\beta}(t_0)$ such that $\vy_i\left(\vx^{\alpha}(t_0),M_z^a,t\right)=\vy_i\left(\vx^{\beta}(t_0),0,t\right)$.
\label{def:undetectable-attack}
\end{definition}
Note that, because of the linearity of~\eqref{eq:lateral-dynamic-output-model}, the attack undetectability condition as presented in Definition~\ref{def:undetectable-attack} is equivalent to finding initial condition $\vx(t_0)$ resulting $\vy_i\left(\vx(t_0),M_z^a, t\right) =\vZero$, specifically 
$\vx(t_0)=\vx^{\alpha}(t_0)-\vx^{\beta}(t_0)$, for the following model:
\begin{equation}
    \begin{aligned}
    \dot{\vx}&=\vA \vx + \vB M_z^a,\\
    \vy_i&=\vC_i\vx.
    \end{aligned}
    \label{eq:only-M-lateral-dynamic-output-model}
\end{equation}
The resulting dynamics of~\eqref{eq:only-M-lateral-dynamic-output-model} after applying the attack signal $M_z^a$, which makes the output equal to zero, is called zero dynamics, and the attack signal is called zero dynamics attacks.
By applying Laplace transformation on~\eqref{eq:only-M-lateral-dynamic-output-model} and setting the output to zero, we get:
\begin{align*}
    s\vx&=\vA\vx+\vB M_z^a,\\
    \vZero&=\vC_i \vx,
\end{align*}
thus:
\begin{equation*}
  \left[\begin{array}{cc}
      s\vI-\vA & -\vB  \\
       \vC  & \vZero 
    \end{array}
    \right] 
    \left[\begin{array}{c}
      \vx \\
    M_z^a
    \end{array}
    \right]=\vZero.
\end{equation*}
The matrix $\vP(s)=\left[\begin{smallmatrix}
      s\vI-\vA & -\vB  \\
       \vC  & \vZero  
\end{smallmatrix}\right]$ is called the Rosenbrock matrix associated with the system~\eqref{eq:only-M-lateral-dynamic-output-model}. The complex values $s_0\in \mathbb{C}$ satisfying $Rank\left(\vP(s_0)\right)< \dim{\left[\begin{smallmatrix}  \vx \\
    M_z^a
\end{smallmatrix}\right]}=3$ are called invariant zeros. 
The following lemma presents necessary and sufficient conditions for the existence of zero dynamics for the model~\eqref{eq:only-M-lateral-dynamic-output-model}.
\begin{lemma}[\cite{pasqualettiControlTheoreticMethodsCyberphysical2015}] 
Consider the linear state space model~\eqref{eq:only-M-lateral-dynamic-output-model}, the system has zero dynamics if and only if it has invariant zeros, i.e. there exist complex value $s_0\in \mathbb{C}$ satisfying $Rank\left(\vP(s_0)\right)<\dim{\left[\begin{smallmatrix}  \vx \\
    M_z^a
\end{smallmatrix}\right]}=3$.
\end{lemma}
\par
\begin{definition}[Disruptive zero dynamics attack~\cite{Shim2022}]
\label{def:disruptive-zero-dynamics}
The zero dynamics attacks are called disruptive if the associated invariant zero is unstable, i.e. the resulting zero dynamics is unstable.
\end{definition}
\subsubsection{Connection with strong observability and detectability properties}
The existence of invariant zeros is related to the properties of strong observability and detectability in linear systems, as explained in the following theorem.
\begin{theorem}[\cite{OvsLetInv}]
   A linear dynamic system is strongly observable if and only if it has no invariant zeros, and it is strongly detectable if and only if all its invariant zeros are stable.
   \label{theo:invariant-zero-stronly}
\end{theorem}

The zero dynamics attacks are injection attacks into the input, and this injection is unknown to the system. 

\textbf{Case 1:} If the system is strongly observable, it means that the system can uniquely reconstruct the state based on the output, without having any information about the unknown attack signal. Therefore, the attacker cannot perform zero dynamics attacks on a strongly observable system. 

\textbf{Case 2:}  If the system is not strongly observable, it means that there can be two different states for the same output, specifically, one state belongs to a system under zero dynamics attacks and one state belongs to an attack-free system. The system which is not strongly observable can be strongly detectable (Case 2.1) or not (Case 2.2).

\textbf{Case 2.1:}  If the system is strongly detectable, the attacked system's state will converge to the attack-free system's state over time. 

\textbf{Case 2.2:}  If the system is not strongly detectable, the attacker can perform zero dynamics attacks that cause the state to diverge while the output is identical to an attack-free system.

\subsection{Problem statement}
The objective of this paper is to answer the following questions: 
\begin{itemize}
\item Do the linear lateral dynamics have invariant zeros? 

\item How can an attacker exploit these invariant zeros to perform zero dynamics attacks? 

\item Are these attacks disruptive? 
\end{itemize}
These questions are answered in Propositions~\ref{prop:1},~\ref{prop:2}, and~\ref{prop3}.
Finally, this paper proposes measures to protect the vehicle's lateral dynamics from zero dynamics attacks.
\section{Zero Dynamics attacks against vehicle's lateral dynamics}
\label{sec:Zero Dynamic attack against vehicle lateral dynamic}
The following three subsections consider the output scenarios: yaw rate, lateral acceleration, and their combination. For each scenario, we examine the existence of invariant zeros and zero dynamics attacks, followed by a discussion analyzing the system under attacks.

\subsection{Exclusive yaw rate output}
\begin{proposition}\label{prop:1}
Consider the linear state space model~\eqref{eq:only-M-lateral-dynamic-output-model}, with an output containing only yaw rate, the system has invariant zero ${s_0=a_{11}}$, and the zero dynamics attacks ${M_z^a=-\frac{a_{21}}{b_2}v_y}$ excite this invariant zero, resulting in stable zero dynamics:
\begin{equation}
    \dot{v}_y=a_{11}v_y.
        \label{eq:zero-dynamic-first-case}
\end{equation}
\end{proposition}
\begin{proof}
For the case where the output is only the yaw rate, the output matrix is given by~\eqref{eq:C1}. The Rosenbrock matrix in this case is given by:
\begin{equation*}
    \left[ 
    \begin{array}{cc}
     s\vI-   \vA & -\vB \\
       \vC_1  & \vZero
    \end{array}
    \right]=
 \left[ 
    \begin{array}{ccc}
s-a_{11}&-a_{12}&0 \\
-a_{21}&s-a_{22}&-b_2\\
0 & 1 & 0
    \end{array}
    \right],
\end{equation*}
which is a square matrix with a determinant of ${b_2(s - a_{11})}$, thus the Rosenbrock matrix is rank-deficient when ${s = a_{11}}$. Therefore, ${s_0 = a_{11}}$ is an invariant zero of the system.
 Note that ${a_{11}=-2 \frac{C_f+C_r}{v_x m}}$  is negative and the invariant zero is stable. Now we find the input that excites the zero dynamics, the output is zero for any ${t>t_0}$, where ${t_0}$ is the onset of the attacks, thus ${\forall t>t_0}$, ${r=0}$ and ${\dot{r}=0}$. Substituting in the dynamics of~\eqref{eq:only-M-lateral-dynamic-output-model} gives:
\begin{align*}
    \left[
    \begin{array}{c}
         \dot{v}_y   \\
        \dot{r}
    \end{array}
    \right] & = \left[\begin{array}{cc}
        a_{11} & a_{12} \\
        a_{21} &  a_{22}
    \end{array} \right] \left[
    \begin{array}{c}
         v_y  \\
         r 
    \end{array}
    \right] + \left[\begin{array}{c}
         0 \\
         b_2 
    \end{array} \right] M_z^a, \\
     \left[
    \begin{array}{c}
         \dot{v}_y  \\
         0 
    \end{array}
    \right] & = \left[\begin{array}{c}
        a_{11} v_y  \\
        a_{21}v_y  +b_2 M_z^a
    \end{array} \right],
\end{align*}
thus the attack signal  ${M_z^a=-\frac{a_{21}}{b_2}v_y}$ excites the invariant zero of the system, and the system dynamics become ${\dot{v}_y=a_{11}v_y}$.
\end{proof}

\textbf{Remark:} The zero dynamics~\eqref{eq:zero-dynamic-first-case} shows that under zero dynamics attacks, the vehicle experiences lateral sliding without any rotational motion. These zero dynamics are stable dynamics, where the invariant zero ${a_{11}}$ is stable. Consequently, based on Theorem~\ref{theo:invariant-zero-stronly}, the system is not strongly observable but is strongly detectable; this means that although the output remains identically zero, the state is not zero but converges to zero over time. As a result, zero dynamics attacks exist, and these attacks are not disruptive, i.e. the resulting zero dynamics are stable.
Although the lateral velocity converges to zero over time, it can still be dangerous for the system to have lateral movement without being observed.  For instance, the system may believe that the vehicle is moving forward, while lateral movement is occurring.

\subsection{Exclusive lateral acceleration output}
\begin{proposition}\label{prop:2}
Consider the linear state space model~\eqref{eq:only-M-lateral-dynamic-output-model}, with an output containing only lateral acceleration, the system has invariant zero ${s_0 = \frac{a_{11}v_x}{a_{12}+v_x}}$, and the zero dynamics attack
\begin{dmath*}
M_z^a=-\frac{1}{b_2(a_{12}+v_x)}\left((a_{11}^2+a_{21}(a_{12}+v_x))v_y\\+(a_{11}a_{12}+(a_{12}+v_x)a_{22})r\right),
\end{dmath*}
excites this invariant zero, resulting in the following zero dynamics:
\begin{equation}
\begin{aligned}
 \dot{v}_y&=\frac{a_{11}v_x}{a_{12}+v_x}v_y=s_0 v_y,\\
 \dot{r}&=\frac{a_{11}v_x}{a_{12}+v_x}r=s_0 r,
    \end{aligned}
    \label{eq:zero-dynamic-second-case}
\end{equation}
which are stable if and only if ${aC_f-bC_r<0}$.
\end{proposition}
\begin{proof}
For the case where the output is only the lateral acceleration, the output matrix is given by~\eqref{eq:C2}. The Rosenbrock matrix becomes in this case:
\begin{equation*}
    \left[ 
    \begin{array}{cc}
     s\vI-   \vA & -\vB \\
       \vC_2  & \vZero
    \end{array}
    \right]=
 \left[ 
    \begin{array}{ccc}
s-a_{11}&-a_{12}&0 \\
-a_{21}&s-a_{22}&-b_2\\
  a_{11}   &
         a_{12}+v_x & 0
    \end{array}
    \right],
\end{equation*}
which is a square matrix with a determinant of ${{b_2\left((a_{12}+v_x)s-a_{11}v_{x}\right)}}$, thus the Rosenbrock matrix is rank-deficient when ${s = \frac{a_{11}v_x}{a_{12}+v_x}}$. Therefore, ${s_0 = \frac{a_{11}v_x}{a_{12}+v_x}}$ is an invariant zero of the system. Substituting ${a_{11}=-2 \frac{C_f+C_r}{v_x m}}$, and ${a_{12}= 2 \frac{b C_r-a C_f}{v_x m}-v_x}$, gives: 
\begin{equation}
    s_0=\frac{C_f+C_r}{a C_f-b C_r} v_x.
    \label{eq:invariant-zero-second-case}
\end{equation}
The sign of ${s_0}$, i.e. the stability of the invariant zero, is determined by the sign of the term ${a C_f-b C_r}$. Generally, the stability of matrix ${\vA}$ does not impose a condition on the sign of this term. The stability condition of the matrix ${\vA}$, ensuring the eigenvalues of ${\vA}$ have negative real parts, is~\cite{gillespie_funamentals_2021,hashemi_comprehensive_2016}:
\begin{equation}
      (a+b)^2-\frac{m(aC_f-bC_r)}{C_r C_f } v_x^2>0.
\label{eq:A-stability-conditions}
\end{equation}
The term ${aC_f-bC_r}$ could be positive while the condition~\eqref{eq:A-stability-conditions} is still satisfied. 

Now we find the input that excites the zero dynamics, the output (the lateral acceleration) is zero thus ${ a_y=a_{11}v_y+(a_{12}+v_x)r=0}$,
and the initial conditions ${r_0}$ and ${v_{y0}}$ should satisfy ${a_{11}v_{y0}+(a_{12}+v_x)r_0=0}$,
and the derivatives ${\dot{v}_y}$ and  ${\dot{r}}$ should  satisfy ${a_{11}\dot{v}+(a_{12}+v_x)\dot{r}=0}$,
substituting the dynamics  ${\dot{v}_y}$ and  $\dot{r}$:
\begin{equation*}
  a_{11}(a_{11}v_y+a_{12}r)+(a_{12}+v_x)(a_{21}v_y+a_{22} r +b_2 M_z^a)=0,
\end{equation*}
thus, the attack signal that excites the invariant zero is:
\begin{dmath*}
    M_z^a=-\frac{1}{b_2(a_{12}+v_x)}\left((a_{11}^2+a_{21}(a_{12}+v_x))v_y\\+(a_{11}a_{12}+(a_{12}+v_x)a_{22})r\right),
\end{dmath*}
and the output is identical to zero ${a_{11}v_y+(a_{12}+v_x)r=0}$, thus ${r=-\frac{a_{11}v_y}{a_{12}+v_x}}$, substituting in the dynamics ${ \dot{v}_y=a_{11}v_y+a_{12}r}$ gives ${ \dot{v}_y=\frac{a_{11}v_x}{a_{12}+v_x}v_y=s_0 v_y}$, thus $\dot{r}=\frac{a_{11}v_x}{a_{12}+v_x}r=s_0 r$, which concludes the proof of the proposition. 
\end{proof}
Fig.~\ref{fig:system-behaviour-under-zero-dynamic-attack} presents the phase plane illustrating the system's behavior under zero dynamics attacks, the state evolves along the zero-output manifold: 1) The state follows the green dashed line when the attacks target the yaw rate, setting it to zero; the invariant zero is stable, and the state converges to zero. 2) The state follows the blue closed-dots line when the attacks target the lateral acceleration, setting it to zero with a stable invariant zero, leading to convergence. 3) The state follows the red dotted line when the attacks target the lateral acceleration, setting it to zero with an unstable invariant zero, causing the state to diverge.
\begin{figure}[htbp]
    \centering
    \includegraphics[width=0.99\linewidth]{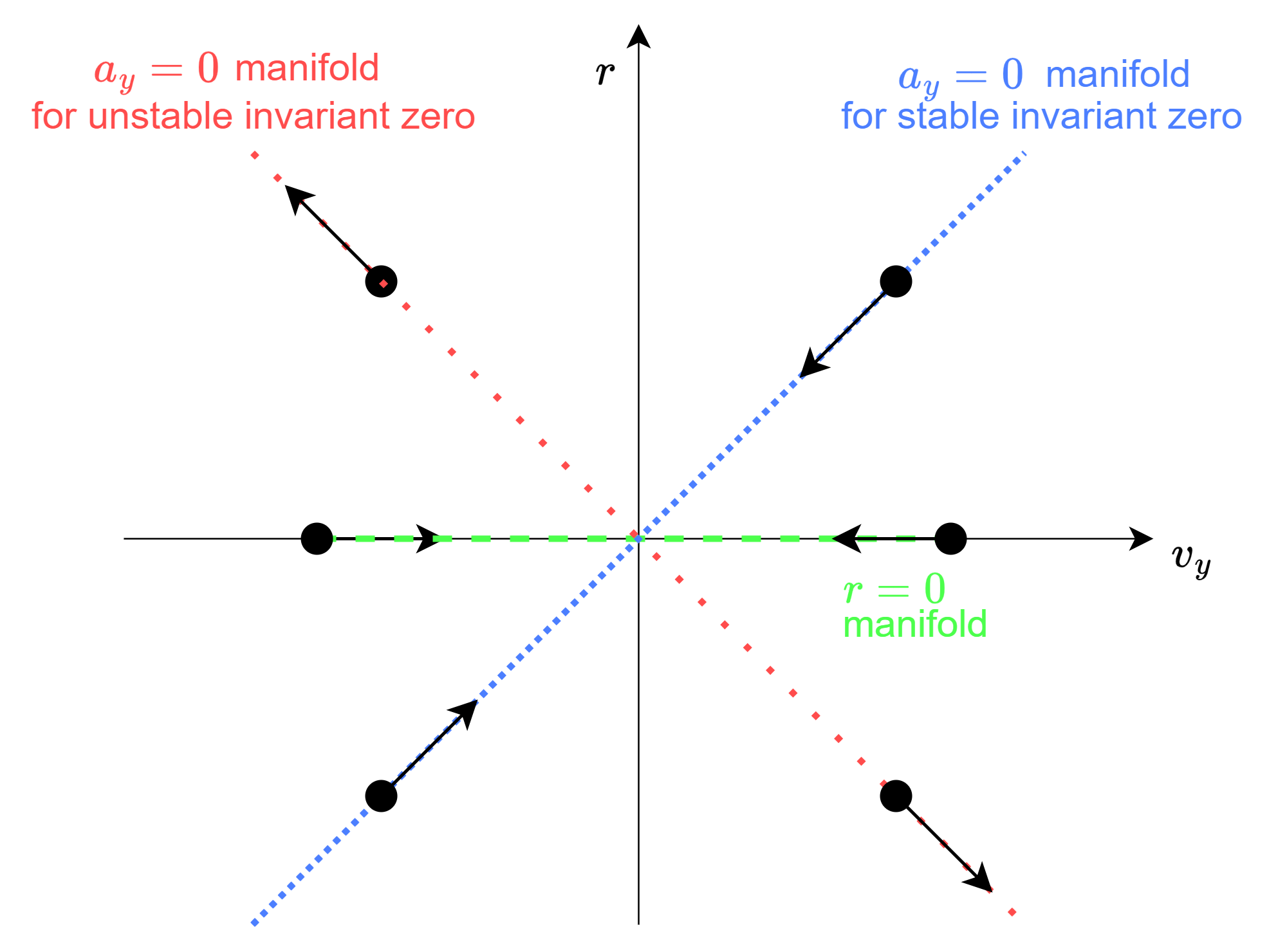}
   \caption{System behavior under zero dynamics attacks for three cases: when the output is the yaw rate (the state lies on the green dashed line), when the output is the lateral acceleration with a stable invariant zero (the state lies on the blue closed-dots line), and when the output is the lateral acceleration with an unstable invariant zero (the state lies on the red dotted line).}
    \label{fig:system-behaviour-under-zero-dynamic-attack}
\end{figure}

\textbf{Remark:}
The zero dynamics~\eqref{eq:zero-dynamic-second-case} show that under zero
dynamics attacks, the vehicle experiences lateral and rotational motion.
According to Proposition~\ref{prop:2},
the invariant zero could be stable or unstable depending on the sign of the term 
${aC_f-bC_r}$. Based on Theorem~\ref{theo:invariant-zero-stronly}, the system is not strongly observable, and it could be not strongly detectable if ${aC_f-bC_r>0}$, in such case, the state diverges to infinity while the output remains at zero. There are threats associated with zero dynamics attacks, and we impose the condition 
\begin{equation}
 aC_f-bC_r<0   
 \label{eq:stable-invariant-zero}
\end{equation}
to prevent them, making the system strongly detectable, and thus ensuring the zero dynamics attacks are not disruptive.

\par
\textbf{Remark, zero dynamics attacks can be detected using lateral and longitudinal acceleration.}

The attacker performs zero dynamics attacks making the lateral acceleration (the output) equal to zero, thus ${a_y=a_{11}v_y+(a_{12}+v_x)r=0}$, while both lateral velocity and yaw rate are nonzero. The longitudinal acceleration, ${a_x}$ is given by ${a_x=\dot{v}_x-v_y r}$, which under the assumption of constant longitudinal velocity becomes ${a_x=-v_y r}$. Note that, for a time window, the longitudinal acceleration can not be equal to zero while both ${v_y}$ and ${r}$ are nonzero. Thus, longitudinal acceleration measurements serve as a detector for zero dynamics attacks that target the lateral acceleration output.

\subsection{Output combines both yaw rate and lateral acceleration}
\begin{proposition}
    \label{prop3}
    Consider the linear state space model~\eqref{eq:only-M-lateral-dynamic-output-model}, with an output that combines both yaw rate and lateral acceleration, the system has no invariant zeros.
\end{proposition}
\begin{proof}
    For the case where the output includes both yaw rate and lateral acceleration, the output matrix is given by~\eqref{eq:C3}. The Rosenbrock matrix becomes in this case:
\begin{equation*}
    \left[ 
    \begin{array}{cc}
     s\vI-   \vA & -\vB \\
       \vC_3  & \vZero
    \end{array}
    \right]=
 \left[ 
    \begin{array}{ccc}
s-a_{11}&-a_{12}&0 \\
-a_{21}&s-a_{22}&-b_2\\
0 & 1 & 0\\
  a_{11}  & a_{12}+v_x & 0
    \end{array}
    \right],
\end{equation*}
the three columns of the Rosenbrock matrix are linearly independent regardless of the value of ${s}$, thus the matrix has a rank of ${3}$ and the system has no invariant zeros. This can be seen also by trying to set the output to zero:
\begin{align*}
    \left[
    \begin{array}{c}
         r  \\
         a_y
    \end{array}
    \right]&= \left[
\begin{array}{cc}
     0 & 1 \\
     a_{11} & a_{12}+v_x 
\end{array}
    \right] \left[
    \begin{array}{c}
         v_y  \\
         r 
    \end{array}
    \right],\\
   \vZero&= \left[
\begin{array}{cc}
     0 \\
     a_{11}  v_y
\end{array}
    \right],
\end{align*}
thus, both ${r}$ and ${v_y}$ (along with their derivatives) are zero, resulting in no zero dynamics.
\end{proof}
\textbf{Remark:} Proposition~\ref{prop3} indicates that there is no threat of zero dynamics attacks when the output combines both yaw rate and lateral acceleration. Based on Theorem~\ref{theo:invariant-zero-stronly}, the system is strongly observable, indicating that the system's states can be fully reconstructed from the output.

\subsection{Summary of findings}
When the output consists of yaw rate, the system exhibits a stable invariant zero, allowing the attacker to perform undetectable but non-disruptive attacks. In the case where the output contains lateral acceleration, the system may have a stable or un unstable invariant zero. A new condition~\eqref{eq:stable-invariant-zero} is proposed to ensure a stable invariant zero. When the invariant zero is unstable, the attacker can perform disruptive and undetectable attacks. Also, this paper proposes to measure the longitudinal acceleration and use it as a detector for zero dynamics attacks. 
Finally, when the output includes both yaw rate and lateral acceleration, the system has no invariant zeros, eliminating the possibility of zero dynamics attacks. This paper recommends using sensors measuring both yaw rate and lateral acceleration to enhance the security of vehicle lateral dynamics against such attacks. Table~\ref{table:summary} shows summarizes these findings.\\
\begin{table}[hbt]
    \caption{Summary of findings. The sign $\checkmark$ means the existence, while $\times$  means the non-existence}
    \centering
\begin{tabular}{|c|c|c|c|}
\hline Output& Additional & Threat of zero & Disruptive  \\
measurements& condition & dynamics attack? &attack? \\
\hline
$r$&- & $\checkmark$ &  $\times$ \\
\hline
$a_y$&$aC_f-bC_r>0$& $\checkmark$ &  $\checkmark$  \\
\hline
$a_y$&$aC_f-bC_r<0$& $\checkmark$ &  $\times$  \\
\hline
$a_y, a_x$&- &  $\times$ &  $\times$ \\
\hline
$r, a_y$&- &  $\times$ &  $\times$ \\
\hline
\end{tabular}
    \label{table:summary}
\end{table}
\section{Simulations}
\label{sec:simulation}
We demonstrate through simulations how zero dynamics attacks can go undetected, leaving no trace in the output. The simulations cover two cases: one where the output is the yaw rate $r$ and another where it is the lateral acceleration $a_y$, as there is no invariant zero when both outputs are combined.
The vehicle's parameters used in these simulations are real parameters belonging to a Sports Utility Vehicle (SUV)~\cite{hashemi_comprehensive_2016}, as shown in Table~\ref{table:parameters}.
\begin{table}[hbt]
    \caption{SUV paremeters}
    \centering
\begin{tabular}{|c|l|l|}
\hline Parameter & Value & Description \\
\hline$m$ & $2270\ (kg)$ & Vehicle mass \\
$I_z$ & $4600\ (kg.m^2)$  & Moment of inertia \\
$a$ & $1.421\ (m)$ & Front axle to CG \\
$b$ &  $1.438\ (m)$  & Rear axle to CG \\
$C_{\alpha_f}$ &$69800\ (N/rad)$  & Front cornering stiffness \\
$C_{\alpha_r}$  & $69600\ (N/rad)$ & Rear cornering stiffness \\
\hline
\end{tabular}
    \label{table:parameters}
\end{table}

\subsection{Exclusive yaw rate output}
Firstly, we consider zero dynamics attacks that aim to maintain $r$ equal to zero, as stated in Proposition~\ref{prop:1}. For this simulation, we assume longitudinal velocity $v_x$ is equal to $25\ m/s$ and that the attacks occur at the initial time when lateral velocity $v_y$ is equal to $5\ m/s$. Fig.~\ref{fig:ZeroOutputCase1} shows $v_y$ and $r$ for a duration of $1$ second. $r$ remains equal to zero while $v_y$ is not. These attacks are undetectable but not disruptive, as $v_y$ converges to zero. 
Secondly, we consider zero dynamics attacks that aim to perform undetectable attacks i.e. the output is identical to the output of an attack-free case, while the lateral velocities have different trajectories. 
For this simulation, we assume $v_x=25\ m/s$ for both the attack and attack-free cases. The attacks occur at the initial time, with the lateral velocity being $5\ m/s$ in the attack case trajectory, and $-5\ m/s$ in the attack-free case trajectory.
We consider the following steering angle input for both attack and attack-free cases:
\begin{equation*}
   \delta(t)= \left\{
\begin{array}{cc}
    0 & t<=0.1, \\
    \sin(10t) & t>0.1.
\end{array}
    \right.
\end{equation*}
Fig.~\ref{fig:TwoTrajectoriesCase1} shows the lateral velocity and the yaw rate for a duration of $1$ second. Both attack-case trajectory and attack-free case trajectory have the same output, but different lateral velocities, however, the attack-case lateral velocity converges to the attack-free one.

\begin{figure}[hbt]
    \centering
    \includegraphics[width=0.99\linewidth]{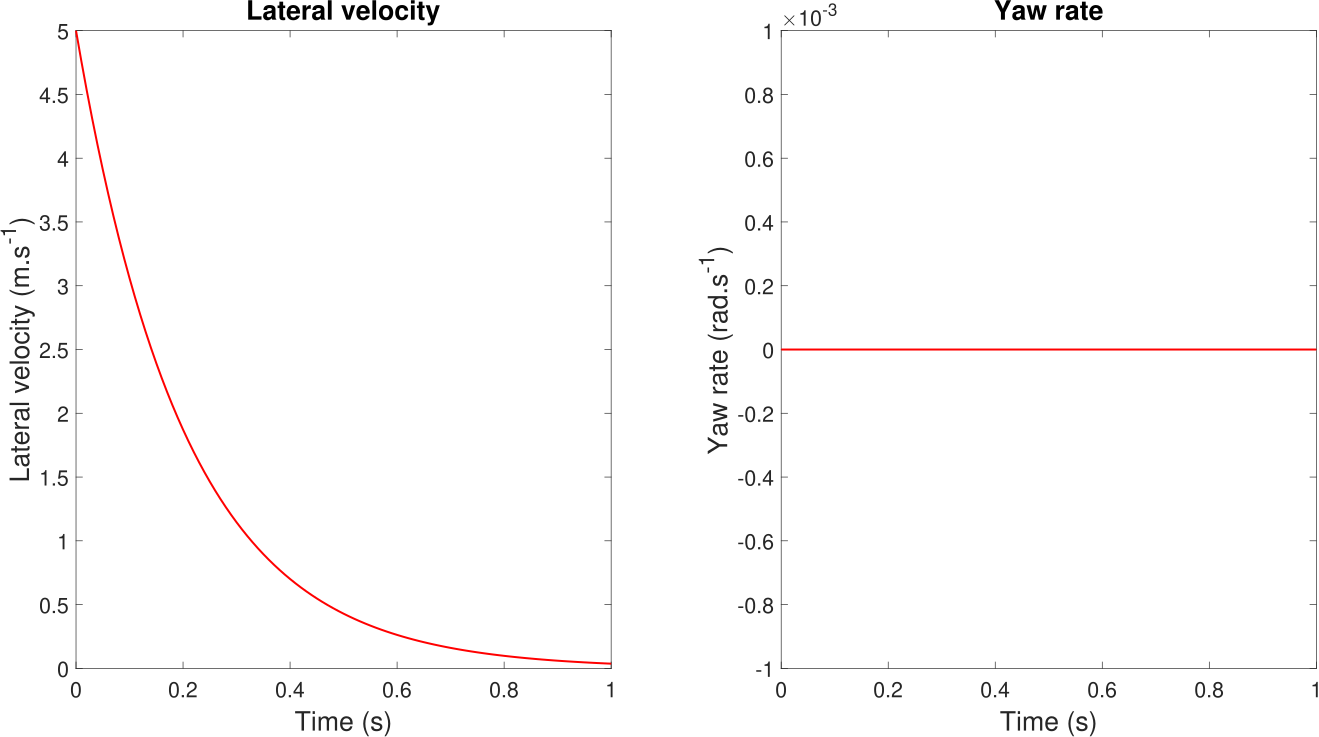}
    \caption{Lateral velocity $v_y$ and yaw rate $r$ when 
the zero dynamics attacks aim to maintain $r$ equal to zero.}
    \label{fig:ZeroOutputCase1}
\end{figure}

\begin{figure}[hbt]
    \centering
    \includegraphics[width=0.99\linewidth]{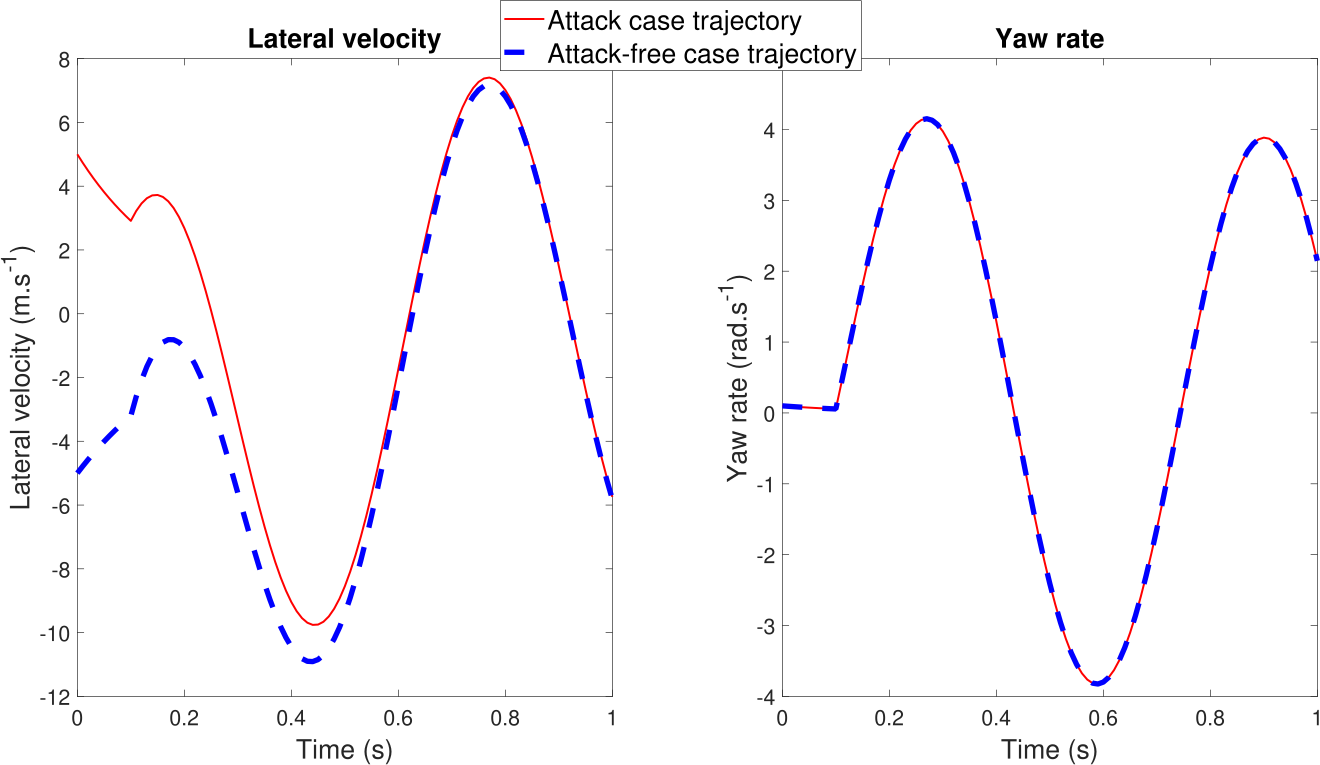}
    \caption{Lateral velocity $v_y$ and yaw rate $r$ when 
the zero dynamics attacks aim to perform undetectable attacks i.e. the output is identical to the one of an attack-free case, while the lateral velocities have different trajectories.}
    \label{fig:TwoTrajectoriesCase1}
\end{figure}
\textbf{Note:} Although the attacks are not disruptive, i.e. the lateral velocity is converging, having a lateral movement that is not observed by the system can still be dangerous in practical situations, on the highway for instance. 
\subsection{Exclusive lateral acceleration output}
\label{subsec:simulation-ay}
We consider zero dynamics attacks that aim to maintain $a_y$ equal to zero, as stated in Proposition~\ref{prop:2}.
The invariant zero, given in~\eqref{eq:invariant-zero-second-case} has large value considering the parameters in Table~\ref{table:parameters}, thus we consider longitudinal velocity $5\ m/s$, and the value of the invariant zero, in this case, is $s_0=-775.3\ 1/s$, which is stable and causes the zero dynamics to converge quickly. 
For this simulation, we assume that the attacks occur at the initial time when the yaw rate is equal to $1 \ rad/s$ and the lateral velocity is equal to $6.4\times 10^{-3} \ m/s$.
Fig.~\ref{fig:ZeroOutputCase2Stable} shows the lateral velocity, yaw rate, and lateral acceleration. The lateral acceleration remains zero while the lateral velocity and yaw rate are not, however, they converge to zero. 

\begin{figure}[hbt]
    \centering
    \includegraphics[width=0.99\linewidth]{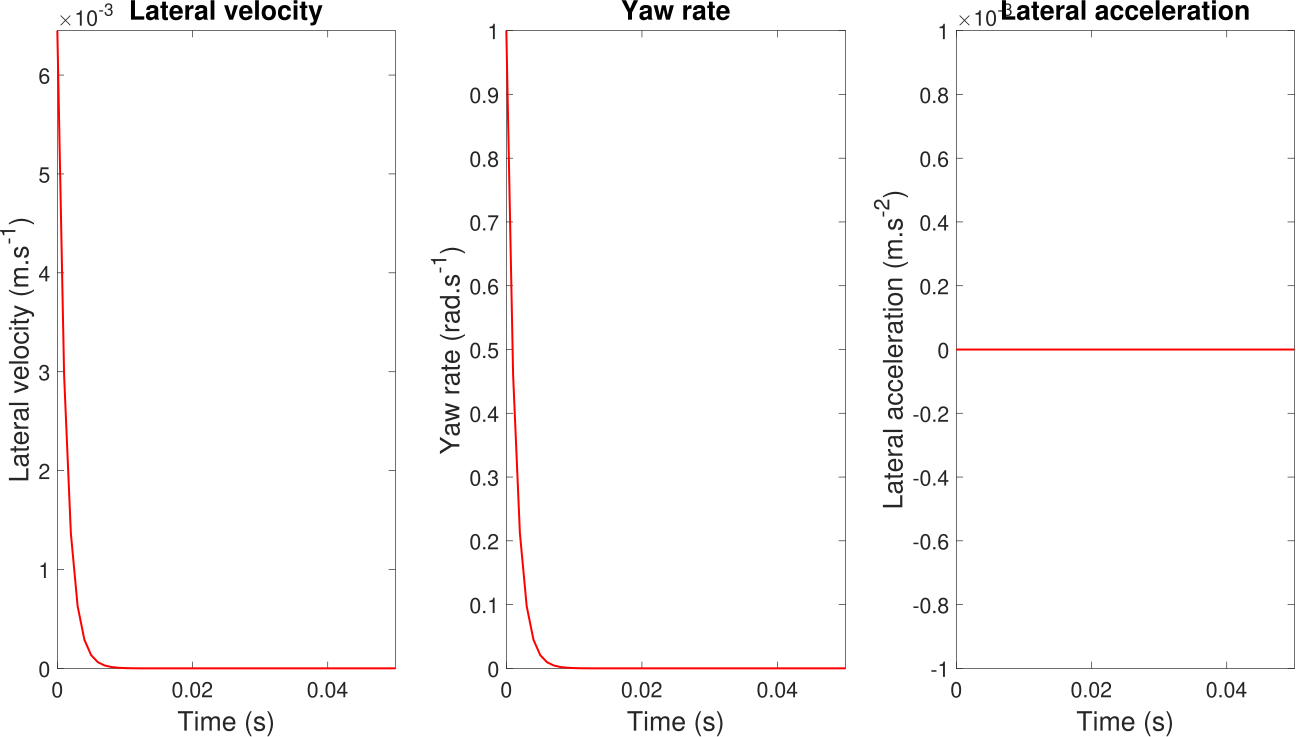}
    \caption{Lateral velocity $v_y$, yaw rate $r$, and lateral acceleration $a_y$ when 
the zero dynamics attacks aim to maintain $a_y$ equal to zero, in the case of stable invariant zero.}
    \label{fig:ZeroOutputCase2Stable}
\end{figure}
Now, we consider a vehicle where the condition~\eqref{eq:stable-invariant-zero} is not satisfied, e.g. the front axle to CG distance is $a=1.521\ m$, the invariant zero will have the value $s_0=114.6\ 1/s$, which is unstable; while the eigenvalues of the matrix $\vA$, in this case, are $(-23.5,
  -27.6)$, indicating its stability. Fig.~\ref{fig:ZeroOutputCase2UnStable} shows that the lateral acceleration remains zero, while the lateral velocity and yaw rate diverge. It is important to mention that, the linear model is not guaranteed to be valid when the vehicle enters the unstable mode.
\begin{figure}[hbt]
    \centering
    \includegraphics[width=0.99\linewidth]{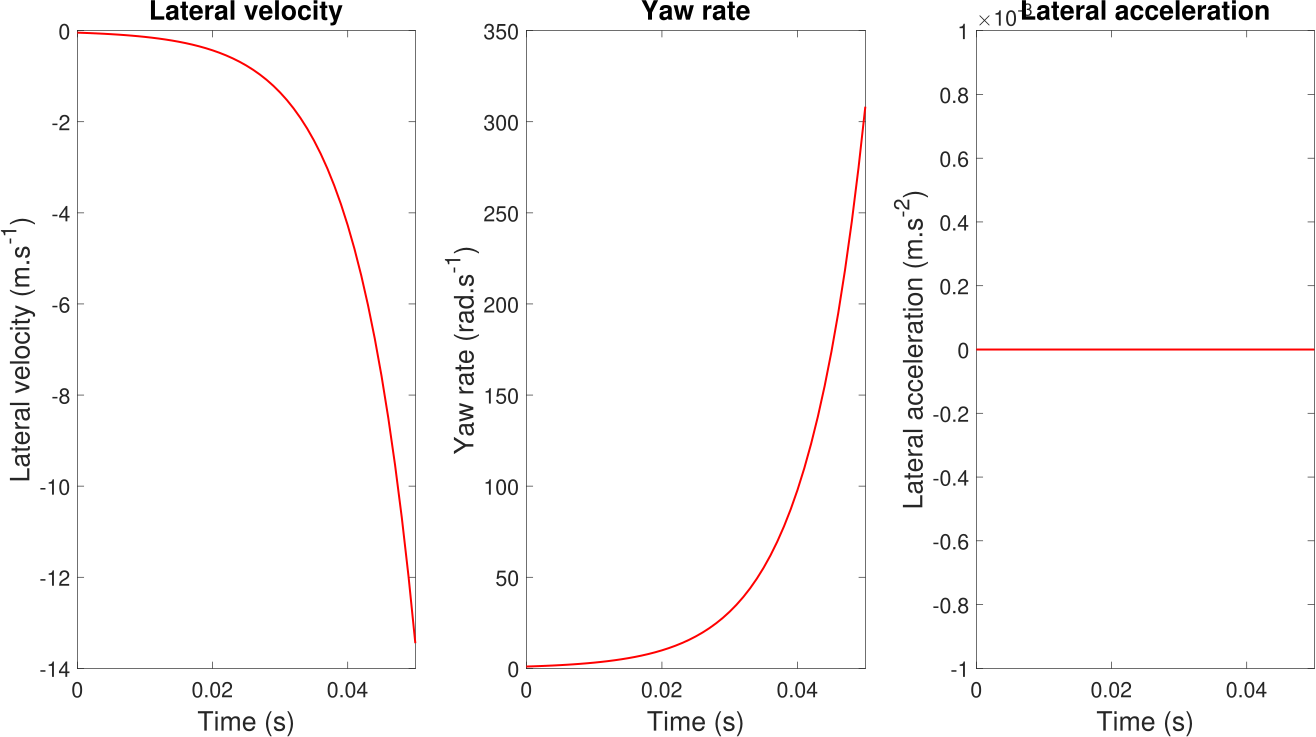}
    \caption{Lateral velocity $v_y$, yaw rate $r$, and lateral acceleration $a_y$ when 
the zero dynamics attacks aim to maintain $a_y$ equal to zero, in the case of unstable invariant zero.}
    \label{fig:ZeroOutputCase2UnStable}
\end{figure}

\section{Conclusion}
\label{sec:conclusion}
This paper studies the linear vehicle lateral dynamics and identifies the three types of outputs. It demonstrates how zero dynamics attacks can exploit the invariant zeros of a linear system to perform undetectable attacks. For each case of the lateral model output, the system’s invariant zeros are studied and analyzed.
This paper recommends that the output consists of both yaw rate and lateral acceleration to prevent zero dynamics attacks, and it recommends using longitudinal acceleration measurements as a detector for zero dynamics attacks when only accelerometers are available.
Future works will consider model nonlinearities and saturations, sensors' noises and more realistic simulations using vehicles' simulators.  
\bibliographystyle{abbrv}
\bibliography{root}
\end{document}